\documentclass[a4paper,UKenglish]{article}
\usepackage{authblk}
\usepackage{a4wide}
\usepackage{url}
\usepackage{hyperref}
\usepackage{amsmath,amssymb,bm,amsfonts,amsthm}
\newtheorem{theorem}{Theorem}

\newtheorem{proposition}{Proposition}
\newtheorem*{claim*}{Claim}
\newtheorem{claim}{Claim}

\usepackage[utf8]{inputenc}
\usepackage{xcolor}
\usepackage{graphicx}
\usepackage{soul}
\usepackage{array,makecell}
\usepackage{multirow}
\usepackage{tikz}
\usetikzlibrary{positioning, arrows.meta}
\usepackage{diagbox}

\definecolor{ForestGreen}{RGB}{34,139,34}

\usepackage[vlined,english,boxed,linesnumbered]{algorithm2e} 
      \SetAlgoLined
      \SetKwInput{Input}{Input}
      \SetKwInput{Output}{Output}
      \SetKw{Return}{Return}
      \SetKwIF{If}{ElseIf}{Else}{If}{then}{else if}{else}{endif}
      \SetKwFor{For}{For}{do}{done}
      \SetKwFor{for}{For}{}{}
      \SetKwFor{While}{While}{do}{done}
      \SetKwFor{Repeat}{Repeat}{}{}
      \SetKwRepeat{Do}{Do}{while}
      \SetKwFor{Procedure}{Procedure}{}{}
      \SetKwFor{Function}{Function}{}{}
      \SetKw{Return}{Return}%
      \SetKw{return}{return}%
      \SetKwComment{Comment}{/* }{ */}
      \SetSideCommentRight
      \DontPrintSemicolon

\newcommand{\set}[1]{\{{#1}\}}

\let\eps=\epsilon

\title{Complexity Gaps between Point and Interval Temporal Graphs for some Reachability Problems\footnote{This work was supported by the French National Research Agency (ANR) through projects ANR-22-CE48-0001 (TEMPOGRAL) and ANR-24-CE48-4377 (GODASse).}} 

\author[1]{Guillaume Aubian}
\affil[1]{IRIF, CNRS \& Université Paris Cité}{}{}{}
\author[2]{Filippo Brunelli}
\affil[2]{European Commission –- Joint Research Centre (JRC)}{}{}{}
\author[3]{Feodor Dragan}
\affil[3]{Kent State University}{}{}{}
\author[4]{Guillaume Ducoffe}
\affil[4]{University of Bucharest, Faculty of Mathematics and Computer Science, and National Institute for Research and Development in Informatics, Romania}{}{}{}
\author[1]{Michel Habib}
\author[1]{Allen Ibiapina}
\author[5]{Laurent Viennot}
\affil[5]{Inria, DI ENS, Paris}{}{}{}

\begin{document}

\maketitle

\begin{abstract}
Temporal graphs arise when modeling interactions that evolve over time. They usually come in several flavors, depending on the number of parameters used to describe the temporal aspects of the interactions: time of appearance, duration, delay of transmission. In the point model, edges appear at specific points in time, whereas in the more general interval model, edges can be present over specific time intervals. In both models, the delay for traversing an edge can change with each edge appearance. When time is discrete, the two models are equivalent in the sense that the presence of an edge during an interval is equivalent to a sequence of point-in-time occurrences of the edge. However, this transformation can drastically change the size of the input and has implications for complexity. 

Indeed, we show a gap between the two models with respect to the complexity of the classical problem of computing a fastest temporal path from a source vertex to a target vertex, i.e. a path where edges can be traversed one after another in time and such that the total duration from source to target is minimized. It can be solved in near-linear time in the point model, while we show that the interval model requires quadratic time under classical assumptions of fine-grained complexity. With respect to linear time, our lower bound implies a factor of the number of vertices, while the best known algorithm has a factor of the number of underlying edges. We also show a similar complexity gap for computing a shortest temporal path, i.e. a temporal path with a minimum number of edges. Here our lower bound matches known upper bounds up to a logarithmic factor. Interestingly, we show that near-linear time for fastest temporal path computation is possible in the interval model when it is restricted to uniform delay zero, i.e., when traversing an edge is instantaneous. However, this special case is not exempt from our lower bound for shortest temporal path computation. These two results should be contrasted with the computation of a foremost temporal path, i.e., a temporal path that arrives as early as possible. It is well known that this computation can be solved in near-linear time in both models.

We also show that there is no gap in testing the all-to-all temporal connectivity of a temporal graph. We demonstrate a quadratic lower bound that applies to both the interval and point models and aligns with the existing upper bounds.
\end{abstract}


\section{Introduction}

Graphs are the standard mathematical framework for representing relationships within networks. However, in many real-world scenarios, such as transportation systems or social networks, these relationships are time-sensitive and evolve over time. Temporal graphs address this dynamic nature by modeling networks where connections change over time.
Their study traces back to time-dependent networks~\cite{CookeH1966} in the context of road networks. They were later re-introduced under various flavors of models, see e.g.~\cite{CasteigtsFQS2012,Holme2015,LatapyVM2018,Michail2016}. 
In the simplest model, each edge appears at specific points in time. 
In the most advanced model, each edge is present for entire time intervals, and the delay for traversing it is given by a time-dependent function, often assumed to be piecewise linear (and, without loss of generality, linear in each interval of appearance).
We are interested in highlighting the difference in complexity between these two models, which we will call \emph{point} temporal graph and \emph{interval} temporal graph, respectively. 

Note that the interval model obviously encompasses the point model, since a point in time corresponds to a time interval with equal bounds. In both models, each appearance of an edge can be represented by a tuple storing the two vertices of the edge with the few associated time parameters (assuming linearity in each appearance interval for the interval model), and the size of the input can be measured by the total number $M$ of such edge appearances. If we remove the time information from the list of tuples describing a temporal graph, and forget about multiplicities, we obtain what we call its underlying graph.
If time is discrete, the two models may seem equivalent, since an edge appearance during an interval of length $\ell$ can be seen as $\ell$ point-by-point appearances. 
However, the size of the input can then grow by an exponential factor, leading to different complexities in solving a problem in one model or the other.

A recent strand of research revisits all classical graph problems in the context of temporal graphs, see e.g.~\cite{AKRIDA2020,CasteigtsPS2019,HAAG2022,KEMPE2002,MertziosMNZZ2023,MICHAIL20162}.
Most of these works focus on the point model while the interval model may have been overlooked. Interestingly, the natural notion of connectivity in temporal graphs arises from temporal paths whose computation is the subject of various works in both models (see e.g. \cite{BuiXuanFJ2003,WuCKHHW2016}). A temporal path is a path whose edges can be traversed one after another sequentially in time. 
%
When considering all temporal paths between a source vertex and a target vertex, the time aspect gives rise to several notions of shortest path: in particular, a \emph{foremost} temporal path arrives at the target as early as possible, a \emph{shortest} temporal path uses a minimum number of edges, and a \emph{fastest} temporal path has a minimum duration, i.e. the time span between leaving the source and arriving at the target is minimal. 

First, we observe discrepancies between the two models with respect to the complexity of known algorithms for computing a shortest or a fastest temporal path. They can both be found in near-linear time in the point model~\cite{WuCKHHW2016}, while in the interval model the best algorithms for computing a shortest temporal path~\cite{BuiXuanFJ2003,jain2022algorithms} or a fastest temporal path~\cite{DehneOS2012} are slower by a factor depending on the size of the input, $n$ and $m=O(n^2)$ respectively, if the input has $n$ vertices and the underlying graph has $m$ edges.
On the other hand, foremost temporal paths can be computed in near-linear time in both models using a variant of Dijkstra algorithm~\cite{Dreyfus1969,BuiXuanFJ2003}.
To the best of our knowledge, no non-trivial lower bounds are known. Thus, we ask whether a slowdown factor proportional to the number of vertices is necessary to compute a fastest or a shortest temporal path in the interval model.

Second, no non-trivial lower bounds seem to be known for the natural problem of temporal connectivity, i.e. the existence of a temporal path between any pair of vertices. Indeed, the best known approach is to perform a single-source foremost temporal path computation for each possible source vertex~\cite{BhadraF2003,AkridaGMS2017}, resulting in time complexity of $\widetilde{\cal O}(nM)$ for a temporal graph with $n$ vertices and $M$ temporal edges in both models. Again, we ask whether a factor proportional to the number of vertices is necessary when testing temporal connectivity, even in a point temporal graph.

The point model has been extensively studied without considering delays (see, e.g., \cite{KEMPE2002,Michail2016}). Temporal paths then classically come in two flavors: \emph{strict} when times of appearance strictly increase along the path, and \emph{non-strict} when times of appearance are non-decreasing. These two flavors are equivalently grasped by delays of one and zero respectively. Similarly, the interval model has been extensively studied without delays in the literature about social networks (see, e.g., \cite{Holme2012,Holme2015,LatapyVM2018}). Therefore, we pay special attention to the specific cases where all the delays are one and zero, respectively.

\begin{table}[t]
  \begin{center}
  \resizebox{\linewidth}{!}{%
  \setlength\extrarowheight{1mm}
  \begin{tabular}{|c|c|c|c|c|}
    \hline
        \multirow{2}{*}{\diagbox[width=25mm]{Problem}{Model}}
        & \multirow{2}{*}{Point}
        & \multirow{2}{*}{Interval}
        & \multicolumn{2}{c|}{Undirected Interval}
        \\
        \cline{4-5}
        &&& delay zero & delay one
        \\
    \hline{}
        Foremost
        & $\widetilde{\cal O}(M)$ \cite{Dreyfus1969,BuiXuanFJ2003}
        & $\widetilde{\cal O}(M)$ \cite{Dreyfus1969,BuiXuanFJ2003}
        & $\widetilde{\cal O}(M)$ \cite{Dreyfus1969,BuiXuanFJ2003}
        & $\widetilde{\cal O}(M)$ \cite{Dreyfus1969,BuiXuanFJ2003}
        \\        
    \hline{}
        Fastest
        & $\widetilde{\cal O}(M)$ \cite{WuCKHHW2016,DibbeltPSW2018}
        & \textcolor{blue}{$\Omega^\times((nM)^{1-\varepsilon})$ \textbf{Th.\ref{thm:hardness-apsp}}}, 
        $\widetilde{\cal O}(n^2M)$ \cite{DehneOS2012}
        & \textcolor{blue}{$\widetilde{\cal O}(M)$ \textbf{Th.\ref{thm:zero-delay}}}
        & \textcolor{blue}{${\Omega^+}((nM)^{1-\eps})$ \textbf{Th.\ref{thm:hardness-triangle}}}
        \\  
    \hline{}
        Shortest
        & $\widetilde{\cal O}(M)$ \cite{WuCKHHW2016}
        & \textcolor{blue}{$\Omega^+((nM)^{1-\varepsilon})$ \textbf{Th.\ref{thm:hardness-shortest}}},
        $\widetilde{\cal O}(nM)$ \cite{BuiXuanFJ2003}
        & \textcolor{blue}{${\Omega^+}((nM)^{1-\eps})$ \textbf{Th.\ref{thm:hardness-shortest}}}
        & \textcolor{blue}{${\Omega^+}((nM)^{1-\eps})$ \textbf{Th.\ref{thm:hardness-shortest}}}
        \\[1mm]
    \hline{}
        Connectivity
        & \textcolor{blue}{${\Omega^*}((nM)^{1-\eps})$ \textbf{Th.\ref{thm:hardness-connectivity}}}
        & $\widetilde{\Theta}(nM)$  \cite{Dreyfus1969,BuiXuanFJ2003}
        & \textcolor{blue}{${\Omega^*}((nM)^{1-\eps})$ \textbf{Th.\ref{thm:hardness-connectivity}}}
        & \textcolor{blue}{${\Omega^*}((nM)^{1-\eps})$ \textbf{Th.\ref{thm:hardness-connectivity}}}
        \\
    \hline
   \end{tabular}
   }
 \end{center}   
  \caption{Summary of known results on the complexity of computing a foremost (resp. fastest, shortest) temporal path or testing temporal connectivity in a temporal graph with $n$ vertices and $M$ temporal edges. 
  The lower-bounds are given for any $\eps>0$ under standard fine-grained complexity assumptions.
  Symbols $\Omega^*, \Omega^\times, \Omega^+$ all stand for $\Omega$ under the following respective assumptions:
  ${\Omega}^*$ holds under SETH; ${\Omega}^\times$ holds under the assumption that there is no algorithm that can solve all pairs shortest path (APSP) in truly subcubic time; and ${\Omega}^+$ applies to combinatorial algorithms assuming that no such algorithm can solve triangle detection in truly subcubic time.
  }
  \label{tbl:results}
\end{table}

\paragraph*{Our contribution}%
We give strong 
evidences towards a positive answer to both of the above questions. Our first main result is a subquadratic reduction from diameter two in graphs to temporal connectivity in point temporal graphs. The former problem consists in testing whether an undirected graph has diameter at most two. It is known to require quadratic time under the Strong Exponential Hypothesis (SETH)~\cite{RodittyV2013}. This reduction implies that, for any $\eps>0$, testing temporal connectivity in point temporal graphs with $n$ vertices and $M$ edge appearances requires $\Omega((nM)^{1-\eps})$ time unless SETH is false.  Our reduction is restricted to \emph{undirected} point temporal graphs with \emph{uniform delay zero}, i.e. each edge can be traversed in both directions with delay zero each time it appears. A similar reduction applies to uniform delay one.
Obviously, this lower bound holds also for the more general setting of interval temporal graphs.

Our second main result is a subcubic reduction from negative triangle detection  to fastest temporal path computation. The former consists in detecting whether a weighted graph has a triangle of negative total edge weight. This reduction implies that, for any $\eps>0$, an $O((nM)^{1-\eps})$-time algorithm for fastest temporal path in interval temporal graphs 
would break the state of the art for several classical problems, including e.g. all pairs shortest path (APSP), the replacement paths problem on weighted directed graphs, and verifying the correctness of a matrix product over the $(\min, +)$-semiring, among others~\cite{WilliamsW2018}. 
Our reduction is restricted to undirected interval temporal graphs with \emph{constant delays}, i.e. each edge can be traversed in both directions with symmetric constant delay in each time interval in which it appears.

We supplement the latter lower bound with the following refinements for uniform delays zero and one.
First, we provide a deterministic ``combinatorial'' near-linear-time algorithm for finding a fastest temporal path in undirected interval temporal graphs with uniform delay zero. 
Our algorithm is a rather simple utilization of a dynamic connectivity algorithm. We indeed solve the slightly more technical \emph{profile} problem which consists in computing, given a source vertex and a target vertex, a representation of the (profile) function that assigns to each possible departure time from the source, the corresponding earliest arrival time at the target. Note that the duration of a fastest temporal path can be inferred from the profile function.
Second, we show that this restricted setting to undirected temporal graphs with uniform delay zero is the widest possible in which we can expect a near-linear time ``combinatorial'' algorithm. More precisely, we provide two subcubic reductions from triangle detection in an undirected graph to fastest temporal path computation.
One lower bound applies to undirected interval graphs with uniform delay one, and the other applies to directed interval graphs with uniform delay zero. Note that triangle detection is a classical problem in fine-grained complexity which is, for example, subcubic equivalent to combinatorial Boolean matrix multiplication~\cite{WilliamsW2018}. The notion of a ``combinatorial'' algorithm has no formal definition. Intuitively, it refers to an algorithm which is efficient both theoretically and practically, i.e., whose complexity has a low leading constant.
Our reductions imply that a truly subcubic combinatorial algorithm for one of these restricted settings would constitute a breakthrough with respect to Boolean matrix multiplication algorithms~\cite{WilliamsW2018}.

Finally, we provide a reduction from triangle detection in undirected graphs to shortest path computation in interval temporal graphs. This leads to a lower bound of $\Omega((nM)^{1-\eps})$ time for combinatorial algorithms.
Interestingly, this reduction uses a temporal graph of diameter $\Theta(n)$, as opposed to our other reductions, which are based on temporal paths of constant length. This is indeed necessary as the algorithm of \cite{BuiXuanFJ2003} finds a shortest temporal path in $\widetilde{\cal O}(Dm)$ where $D\le n$ is the largest length of a temporal path (in terms of number of edges) and $m$ is the size of the underlying graph. Note that it uses a sorted data structure which can be pre-computed in $\widetilde{\cal O}(M)$ time. Here, our reduction applies to undirected temporal graphs with uniform delay zero. A slight variant applies to those with uniform delay one. We thus do not expect a faster algorithm in these cases.

See Table~\ref{tbl:results} for a summary of our results and previously known results. Note that our lower bounds for temporal connectivity and shortest temporal path computation are tight with respect to known upper bounds.

\subparagraph*{More related work.}
The difference in complexity between the interval model and the point model has already been noticed in~\cite{BrunelliCV2021} about algorithms solving the profile problem. 
Profile computation cannot be performed in polynomial time if the delay for traversing a temporal edge is linear~\cite{FoschiniHS2014}, but a fastest temporal path can still be computed in $\widetilde{\cal O}(mM)$ time by combining~\cite{FoschiniHS2014} and~\cite{DehneOS2012}. 
The possibility of using dynamic connectivity for computing (instantaneous) strongly connected components is discussed in \cite{RannouML2020}. The following complexity gap between the point and interval models is shown in \cite{CauviV2025}: Finding a restless temporal path (where waiting is not allowed) in a point temporal graph is fixed-parameter tractable (FPT) when parameterized by the vertex-interval-membership-width, a temporal graph parameter introduced in \cite{BumpusBM2023}. However, it is NP-hard in interval temporal graphs of vertex-interval-membership-width equal to three.

\paragraph*{Structure of the paper}

After formally introducing temporal graphs, we first prove lower bounds for temporal connectivity and shortest temporal path computation. We then study fastest temporal path computation, first proving several lower bounds and then presenting an algorithm that escapes these bounds by restricting the input to undirected edges and zero delays. 

\section{Preliminaries}
We represent times with integers\footnote{This restriction does not weaken our lower bounds. Our upper bound algorithm easily generalizes to any representation of reals.} and use Greek letters to name times. An \emph{interval} of time $[\tau_1,\tau_2]$ with $\tau_1\le \tau_2$ represents all times $\tau$ satisfying $\tau_1\le \tau\le \tau_2$. For $\tau_1<\tau_2$, we similarly define $(\tau_1,\tau_2]$ ($[\tau_1,\tau_2)$ respectively) as the set of times $\tau$ satisfying $\tau_1 < \tau\le \tau_2$ ($\tau_1 \le \tau < \tau_2$ respectively).

An \emph{interval temporal graph}, hereafter referred to simply as a temporal graph, is a pair $G = (V, E)$, where $V$ is a set of vertices and $E$ is a set of \emph{temporal edges}. Each temporal edge is a quintuple $e = (u, v, \tau_1, \tau_2, \delta)$, where $u$ and $v$ are vertices, $\tau_1$, $\tau_2$ and $\delta$ are integers,
with $\tau_1\le \tau_2$ and $\delta\ge 0$ that represent its \emph{begin time}, its \emph{ending time} and its \emph{delay} (or traveling time), respectively. Such a temporal edge represents the presence of edge $uv$ during $[\tau_1,\tau_2]$, that is, it can be traversed starting from its \emph{tail} $u$ at any time $\tau\in [\tau_1,\tau_2]$ to arrive in its \emph{head} $v$ at time $\tau+\delta$. We say that $e$ connects $uv$ during interval $[\tau_1,\tau_2]$ with delay $\delta$.
We let $tail(e)=u$, $head(e)=v$, $begin(e)=\tau_1$, $end(e)=\tau_2$, and $delay(e)=\delta$ denote its tail, its head, its begin time, its ending time, and its delay respectively.
When $\tau_1=\tau_2$ the interval is reduced to a point we refer to the \emph{point model}.  
We mostly consider \emph{undirected} temporal graphs where edges can be traversed in both directions. In that case, it is assumed that $E$ is \emph{symmetric}: for each $(u, v, \tau_1, \tau_2, \delta)\in E$, the symmetrical temporal edge $(v, u, \tau_1, \tau_2, \delta)$ is also in $E$. In other words, $G$ is said to be undirected when $E$ is symmetric.
$G$ is said to have \emph{uniform delay $\delta$} if all temporal edges have same delay $\delta$.


The underlying graph of $G$ is the (static) directed graph $G'=(V,E')$ with same vertex set $V$ and edge set $E'=\set{uv : \exists (u, v, \tau_1, \tau_2, \delta)\in E}$. Note that we let $uv$ denote the directed edge from $u$ to $v$. We say that an edge $uv$ of the underlying graph is \emph{present} during $[\tau_1,\tau_2]$ with delay $\delta$ when there is a temporal edge $(u,v,\tau_1,\tau_2,\delta)$ in $E$.  When the temporal graph is undirected we also assume that the symmetrical edge is in $E$ so that $E'$ is symmetric. When $\tau_1=\tau_2$, we say that $uv$ is present at \emph{point} $\tau_1$ with delay $\delta$.
Note that the underlying graph $G'$ is directed. 
In contrast, when considering a (static) undirected graph $(V,E)$ with vertex set $V$ and edge set $E$, we consider that $E\subseteq \binom{V}{2}$ is a set of unordered pairs, and we let $\{u,v\}$ denote an undirected edge between $u,v\in V$.

When a temporal graph is clear from the context, we let $n=|V|$ denote its number of vertices, $M=|E|$ denote its number of temporal edges, and $m=|E'|$ denote its number of underlying edges. It should be noticed that $M$ can be much larger than $m$, since an edge in $E'$ may correspond to many temporal edges in $E$. For simplicity, we assume that any vertex appears at least in one temporal edge so that we have $n=O(m)$.

A \emph{temporal walk} in $G$ is a sequence of pairs $Q=(e_1, \tau_1), \ldots, (e_k, \tau_k)$ where, for each $i \in [k]=\set{1,\ldots,k}$, $e_i\in E$ is a temporal edge and $\tau_i$ is a traversal time in the presence interval of $e_i$, i.e. satisfying $begin(e_i)\le \tau_i\le end(e_i)$.
It is also required that $head(e_i)=tail(e_{i+1})$ for all $i\in [k-1]$, i.e., it induces a walk in the underlying graph $G'$, and that $\tau_i+delay(e_i)\le \tau_{i+1}$ for all $i\in [k-1]$, i.e., each temporal edge is traversed one after the other in time.
When $s=tail(e_1)$ and $head(e_k)=t$, we say that $Q$ is a \emph{temporal $st$-walk}. We also say that $Q$ is a temporal walk departing from $s$ at time $\tau_1$ and arriving in $t$ at time $\tau_k+\delta_k$ where $\delta_k=delay(e_k)$. We let $dep(Q)=\tau_1$,  $arr(Q)=\tau_k+\delta_k$ and $len(Q)=k$ denote its \emph{departure time}, its \emph{arrival time} and its \emph{length} in number of edges respectively. Its \emph{duration} is defined as $dur(Q)=arr(Q)-dep(Q)=\tau_k+\delta_k-\tau_1$. $Q$ is said to be a \emph{loop} when $tail(e_1)=head(e_k)$.
When all vertices $tail(e_1),head(e_1),\ldots,head(e_k)$ are pairwise distinct, $Q$ is said to be a \emph{temporal path}, or a \emph{temporal $st$-path} if $s=tail(e_1)$ and $head(e_k)=t$. Note that any temporal $st$-walk with $s\not=t$ can be transformed into a temporal $st$-path by removing sequences of temporal edges forming loops.
A \emph{fastest} temporal $st$-path is defined as a temporal $st$-path with minimum duration.
A \emph{shortest} temporal $st$-path is defined as a temporal $st$-path with minimum length.
A \emph{foremost} temporal $st$-path is defined as a temporal $st$-path with minimum arrival time.
If there is a temporal $st$-path for any vertices $s, t$, $G$ is \emph{temporally connected}. 

A walk $W$ in the underlying graph $G'$ will be given as a sequence $v_0,\ldots,v_k$ of vertices such that $v_{i-1}v_i$ is an edge for each $i\in[k]$. It is said to be \emph{temporally feasible} if there exists a sequence of times $\tau_1\le \cdots\le\tau_k$ such that each edge $v_{i-1}v_i$ is present at time $\tau_i$ with delay at most $\tau_{i+1}-\tau_i$ (using $\tau_{k+1}=\infty$). It is said to be a path when $v_0,\ldots,v_k$ are pairwise distinct.

\section{Lower bound for temporal connectivity}

Recall that an interval temporal graph is temporally connected if it has a temporal $st$-path for any pair $s,t$ of vertices.
Since computing the set of temporally reachable vertices from a given vertex can be done in linear time up to a logarithmic factor~\cite{Dreyfus1969,BuiXuanFJ2003}, there is an $\widetilde{\cal O}(nM)$-time algorithm for testing whether an interval temporal graph is temporally connected. 
We now state that this bound is tight under SETH, even in point temporal graphs.

\begin{theorem}\label{thm:hardness-connectivity}
    Assuming {\sc SETH}, for any $\varepsilon > 0$, there is no ${\cal O}((nM)^{1-\varepsilon})$-time algorithm for testing whether a temporal graph is temporally connected, even if the temporal graph is an undirected point temporal graph with uniform delay zero (resp. one).
\end{theorem}

\begin{proof}
    Assuming SETH, for any $\eps>0$, there is no $O(n^{2-\eps})$-time algorithm distinguishing diameter $2$ from $3$ in graphs with $O(n)$ nodes and $O(n)$ edges~\cite{Vassilevska2018}. To prove our result, it thus suffices to reduce {\sc Diameter-Two} to \textsc{Temporal-Connectivity}, where the {\sc Diameter-Two} problem asks whether a (static) undirected graph has diameter at most two, and the \textsc{Temporal-Connectivity} problem asks whether a temporal graph is temporally connected.
    For that, let $G=(V,E)$ be an arbitrary undirected graph with $O(n)$ nodes and $O(n)$ edges.
    Without loss of generality, $G$ is connected and has no self-loop. Let $F$ be the set of ordered pairs $(u, v) \in V^2$ such that $\{u,v\} \in E$ (thus $|F| = 2|E|$).
    We construct from $G$ an undirected point temporal graph $H_G$ as follows:
    \begin{itemize}
        \item  The vertex set is $\{s\} \cup V \cup F_1 \cup F_2$ where $F_1, F_2$ are disjoint copies of $F$ and $s$ is an additional vertex. 
        For every \emph{ordered} pair $f = (u, v) \in F$, we let $f_1 \in F_1, \ f_2 \in F_2$ denote its copies. $F_1 \cup F_2$ induces an independent set in the underlying graph of $H_G$ and so does $\{s\}\cup V$.
        
        
        \item For every ordered pair $f = (u, v) \in F$, we add temporal edges with delay zero in $H_G$ so that:
        \begin{itemize}
            \item edge $uf_1$ is present at time $1$ (i.e. interval $[1,1]$),
            \item edge $f_1v$ is present at time $2$,
            \item edge $uf_2$ is present at time $3$,
            \item edge $f_2v$ is present at time $4$,
            \item edges $sf_1$ and $sf_2$ are present at times $0$ and $5$.
        \end{itemize}
    \end{itemize}
    We also add the corresponding symmetrical temporal edges so that $H_G$ is undirected. All temporal edges of $H_G$ have delay zero.  See also Figure~\ref{fig:lower_bound_connectivity}.
    Note that $H_G$ has $O(n)$ vertices and $M=O(n)$ temporal edges as $G$ has $O(n)$ vertices and $O(n)$ edges.
    
    \begin{figure}[!ht]
    \centering
    \begin{tikzpicture}[scale=1, every node/.style={circle, draw, minimum size=8mm, font=\small}, >=Stealth]

        \node (s) at (3, 0) {$s$};
        \node (u) at (0, 0) {$u$};
        \node (f1) at (3,2) {$f_1$};
        \node (f2) at (3,-2) {$f_2$};
        \node (v) at (6, 0) {$v$};

        \draw[thick] (u) -- node[midway, below, draw=none] {$1$} (f1);
        \draw[thick] (u) -- node[midway, below, draw=none] {$3$} (f2);
        \draw[thick] (v) -- node[midway, below, draw=none] {$2$} (f1);
        \draw[thick] (v) -- node[midway, below, draw=none] {$4$} (f2);
        \draw[thick] (s) -- node[midway, right, draw=none] {$0,5$} (f1);
        \draw[thick] (s) -- node[midway, right, draw=none] {$0,5$} (f2);

    \end{tikzpicture}
    \caption{Part of the point temporal graph $H_G$ induced by nodes $s,u,v,f_1,f_2$ for an ordered pair $f = (u, v)$ such that $G$ contains the undirected edge $\{u,v\}$. Each undirected edge $\{x,y\}$ with label $\tau$ in this Figure represents two point temporal edges with delay zero: $(x,y,\tau,\tau,0)$ and $(y,x,\tau,\tau,0)$. Note that $H_G$ contains a similar structure with two other vertices $f_1'\in F_1$ and $f_2'\in F_2$ for the pair $(v,u)$.}
    \label{fig:lower_bound_connectivity}
\end{figure}
    
    Since $\{s\} \cup F_1 \cup F_2$ induces a temporal graph, whose underlying graph is connected, and where all edges are present at time $0$, vertices in $\{s\} \cup F_1 \cup F_2$ can reach each other at time $0$. Similarly, they can reach each other at time $5$. Since there is a temporal edge between every vertex in $V$ and some vertex in $\{s\} \cup F_1 \cup F_2$ at times $1$ and $4$, this implies that $H_G$ is temporally connected if and only if the vertices in $V$ can temporally reach each other.
    
    Let $u, v$ be two distinct vertices of $V$. We will prove that there is a path of length at most $2$ from $u$ to $v$ in $G$ if and only if there is a temporal path from $u$ to $v$ in $H_G$.
    Let us first prove the direct implication. 
    If $u,w,v$ is a path of length $2$ in $G$, let $f = (u, w)$ and $g = (w, v)$. Then $u,f_1,w,g_2,v$ is a path from $u$ to $v$ in the underlying graph of $H_G$ which is temporally feasible in $H_G$ by traversing edges $uf_1, f_1w, wg_2,g_2v$ at times $1,2,3,4$, respectively.
    
    Let us now prove the converse implication, and suppose that there is a temporal path $P$ from $u$ to $v$ in $H_G$. Since any edge incident to $u$ or $v$ is only present at times $1$, $2$, $3$, or $4$, $P$ only uses times $1$, $2$, $3$, and $4$. This implies that $P$ does not use $s$ only available at times $0$ or $5$.
    Since the underlying graph of $H_G$ is bipartite, with bipartition $(V \cup \{s\}, F_1 \cup F_2)$, and 
    since any two edges $xf,fy$ incident to $f \in F_1 \cup F_2$, with $x\not=y$, have distinct time availabilities, 
    this implies that $P$ uses at most $3$ elements of $F_1 \cup F_2$, i.e., the underlying path of $P$ is either $u,f,v$, or $u,f,w,f',v$, or $u,f,w,f',x,f'',v$, with $w,x \in V$ and $f, f',f'' \in F_1 \cup F_2$. {The latter case is indeed impossible. The distinct time availabilities of $uf$ and $fw$ imply that $fw$ is traversed at time $2$ at least. Hence, both $wf'$ and $f'x$ are traversed at time $2$ or later, implying $f'\in F_2$. Edges $wf'$ and $f'x$ are thus traversed at time $3$ and $4$, respectively, in contradiction with the two distinct time availabilities of $xf''$ and $f''v$.} By construction, the two first cases implies that either $\{u,v\} \in E$ or that $\{u,w\}, \{w,v\} \in E$.  Thus, there is a path of length at most $2$ between $u$ and $v$ in $G$.
    
    We thus conclude that $G$ has diameter at most $2$ if and only if $H_G$ is temporally connected.
    
    \medskip
    To obtain the result with uniform delay one, we modify the above construction so that all temporal edges have delay one and edges involving $s$ also appear at times $-1$ and $6$.
    More precisely, we obtain $H_G'$ from $H_G$ by changing all delays to one instead of zero and by adding edges $sf$ at times $-1$ and $6$ for all $f\in F_1\cup F_2$ (again with delay one and adding corresponding symmetrical temporal edges so that $H_G'$ is undirected).
    
    Now, vertices in $\{s\}\cup F_1\cup F_2$ can reach each other at time 1 by passing through $s$. They can also reach each other starting at time 5 and arriving at time 7. They can then reach any $v\in V$ at time 2. Similarly, any vertex $v$ can reach some $f\in F_1\cup F_2$ at time 2 and then any other $f'\in \{s\}\cup F_1\cup F_2$ at time 7. This implies that $H_G'$ is temporally connected if and only if all vertices in $V$ can temporally reach each other. The proof that this happens if and only if $G$ has diameter two is the same as above since the proof uses only strict temporal paths that also correspond to temporal paths in $H_G'$.
    %
\end{proof}

\section{Lower bound for shortest temporal path}

\begin{theorem}\label{thm:hardness-shortest}
    If, for every $\varepsilon > 0$, there is no combinatorial ${\cal O}(n^{3-\varepsilon})$-time algorithm for detecting a triangle in an undirected graph, then for every $\eps' > 0$, there is no combinatorial ${\cal O}((nM)^{1-\eps'})$-time algorithm for computing a shortest temporal $st$-path in a temporal graph, even if the temporal graph is undirected and has uniform delay zero (resp. one).
\end{theorem}

\begin{proof}
Let $G = (V,E)$ be an instance of \textsc{Triangle Detection} that is a (static) undirected graph. Triangle detection consists in finding three pairwise distinct vertices $u,v,w$ such that the three undirected edges $\{u,v\}, \{v,w\}, \{w,u\}$ are in $E$. Suppose $V = \{v_1,\ldots, v_n\}$. A temporal graph $H_G$ is constructed as follows. Here, a \emph{permanent temporal edge} between vertices $u$ and $v$ indicates that both edges $uv$ and $vu$ are present during the interval $[1,n]$. All delays are zero.
\begin{itemize}
    \item The vertex set of $H_G$ consists of the vertices $s$, $t$, and $4$ disjoint copies of $V$, denoted $V_1,\ldots,V_4$. For $j \in \{1, 2, 3, 4\}$ and $i \in \{1,\ldots,n\}$, let $v_i^j$ represent vertex $v_i$ in $V_j$.
    \item In $V_1$ and $V_4$, permanent temporal edges are defined to form a path that respects the vertex ordering. Specifically, for $j \in \{1, 4\}$ and $i \in \{1,\ldots,n-1\}$, we add a permanent temporal edge between $v_i^j$ and $v_{i+1}^j$. Additionally, permanent temporal edges are defined to connect the source vertex $s$ to $v_n^1$ and the target vertex $t$ to $v_1^4$.
    \item For each edge $\{v_i, v_j\} \in E$ with $i < j$, the temporal edges $v_i^1 v_j^2$ and $v_j^3 v_i^4$ (and their symmetric edges $v_j^2 v_i^1, v_i^4 v_j^3$) are both present at time $i$.
    Additionally, a permanent temporal edge between $v_i^2$ and  $v_j^3$ is added.
\end{itemize}
Observe that $H_G$ can be constructed in ${\cal O}(n+m)$ time, where $m$ is the number of edges of $G$.
The correctness of the reduction comes from the following claim. The theorem thus follows from its proof. 

\begin{claim*}
There is a triangle in $G$ if and only if there exists a temporal path in $H_G$ from $s$ to $t$ using at most $n+4$ edges.
\end{claim*}

First, suppose that there is a triangle in $G$ formed by the vertices $v_{i_1}, v_{i_2}, v_{i_3}$ with $i_1 < i_2 < i_3$. Using the permanent edges within $V_1$, we can construct a temporal path starting at $s$ and arriving at $v_{i_1}^1$ at time $1\le i_1$. Let $P_1$ denote this temporal path and note that it contains $n+1-i_1$ edges. By the existence of the triangle, we know that the temporal edges 
$e_1=(v_{i_1}^1, v_{i_2}^2, i_1, i_1, 0), e_2=(v_{i_2}^2, v_{i_3}^3, 1, n, 0), e_3=(v_{i_3}^3, v_{i_1}^4, i_1, i_1, 0)$
 form a temporal path $(e_1,i_1), (e_2, i_1), (e_3,i_1)$ in $H_G$. Denote this temporal path as $P_2$. Using the permanent edges within $V_4$, we can construct a temporal path $P_3$ from $v_{i_1}^4$ to $t$ that uses $i_1$ edges and traverses all of them at time $n\ge i_1$.  The paths $P_1$, $P_2$, and $P_3$ are constructed such that they can be concatenated into a single temporal path from $s$ to $t$. The total number of edges in this concatenated path is $n+4$. See Figure~\ref{fig:lower_bound_shortest} for an illustration of such a temporal path.

\begin{figure}[!ht]
    \centering
    \begin{tikzpicture}[scale=1, every node/.style={circle, draw, minimum size=7mm}, >=Stealth]

        \node (s) at (-1, 4) {$s$};
        \node (t) at (9, 4) {$t$};

        \node (v10) at (1, 6) {$v_1^1$};
        \node (v11) at (1, 4.5) {$v_{i_1}^1$};
        \node (v12) at (1, 3) {$v_{i_2}^1$};
        \node (v13) at (1, 1.5) {$v_{i_3}^1$};
        \node (v14) at (1, 0) {$v_5^1$};

        \node (v20) at (3, 6) {$v_1^2$};
        \node (v21) at (3, 4.5) {$v_{i_1}^2$};
        \node (v22) at (3, 3) {$v_{i_2}^2$};
        \node (v23) at (3, 1.5) {$v_{i_3}^2$};
        \node (v24) at (3, 0) {$v_5^2$};

        \node (v30) at (5, 6) {$v_1^3$};
        \node (v31) at (5, 4.5) {$v_{i_1}^3$};
        \node (v32) at (5, 3) {$v_{i_2}^3$};
        \node (v33) at (5, 1.5) {$v_{i_3}^3$};
        \node (v34) at (5, 0) {$v_5^3$};

        \node (v40) at (7, 6) {$v_1^4$};
        \node (v41) at (7, 4.5) {$v_{i_1}^4$};
        \node (v42) at (7, 3) {$v_{i_2}^4$};
        \node (v43) at (7, 1.5) {$v_{i_3}^4$};
        \node (v44) at (7, 0) {$v_5^4$};

        \draw[thick, dotted, rounded corners, blue!50] (0.5, 6.5) rectangle (1.5, -0.5);
        \draw[thick, dotted, rounded corners, blue!50] (2.5, 6.5) rectangle (3.5, -0.5);
        \draw[thick, dotted, rounded corners, blue!50] (4.5, 6.5) rectangle (5.5, -0.5);
        \draw[thick, dotted, rounded corners, blue!50] (6.5, 6.5) rectangle (7.5, -0.5);

        \node[above=0.2cm, font=\large, text=blue!50, draw=none] at (1, 6.5) {$V_1$};
        \node[above=0.2cm, font=\large, text=blue!50, draw=none] at (3, 6.5) {$V_2$};
        \node[above=0.2cm, font=\large, text=blue!50, draw=none] at (5, 6.5) {$V_3$};
        \node[above=0.2cm, font=\large, text=blue!50, draw=none] at (7, 6.5) {$V_4$};

        
        \foreach \i/\j in {1/2, 2/3, 3/4} {
            \draw[thick, blue,<-] (v1\i) -- node[midway, right, draw=none] {$1$} (v1\j);
        }
        \foreach \i/\j in {0/1} {
            \draw[thick, blue, <-] (v4\i) -- 
            node[midway, left, draw=none] {$n$}
            (v4\j);
        }

        \draw[thick, blue, ->] (s) -- node[midway, left, draw=none] {$1$} (v14);
        \draw[thick, blue, <-] (t) -- node[midway, right, draw=none] {$n$} (v40);

        \draw[thick, dashed, red, ->] (v11) -- node[shift={(-0.3,0.3)}, midway, right, draw=none] {$i_1$} (v22);
        \draw[thick, blue, ->] (v22) -- node[midway, above, draw=none] {$i_1$} (v33);
        \draw[thick, dashed, red, ->] (v33) -- node[midway, above, draw=none] {$i_1$} (v41);

    \end{tikzpicture}
    \caption{A temporal path in $H_G$ corresponding to a triangle $v_{i_1},v_{i_2},v_{i_3}$ in $G$, assuming that $G$ has $n=5$ vertices with $i_1=2,i_2=3,i_3=4$. Labels indicate when edges are traversed. Plain blue edges correspond to permanent temporal edges, i.e., they are present during interval $[1,n]$, while dashed red edges are present only at point $i_1$.}
    \label{fig:lower_bound_shortest}
\end{figure}

Now, suppose that there exists a temporal path $P$ from $s$ to $t$ of length at most $n+4$. Let $P'$ be the temporal path obtained from $P$ by removing its first and last edges (those incident to $s$ and $t$ respectively). 
Let $e$ be the first temporal edge of $P'$ that is not within $V_1$, and let $g$ be the last temporal edge of $P'$ that is not within $V_4$. 
By the structure of $H_G$, $e$ connects a vertex in $V_1$, say $v_{i_1}^1$, with a vertex in $V_2$, say $v_{i_2}^2$, and it is present only at time $i_1$. Similarly, $g$ is incident to vertices $v_{i_3}^3$ and $v_{i_4}^4$, and it is present only at time $i_4$. Moreover, by the definition of temporal path, $g$ appears after (or at the same time as) $e$ in $P$ (possibly in a non-strict path they appear at the same time). Thus, $i_1 \leq i_4$.  
Now, observe that the distance from $s$ to $v_{i_1}^1$ in the underlying graph of $H_G$ is $n - i_1 +1$, the distance from $v_{i_1}^1$ to $v_{i_4}^4$ is at least 3, and the distance from $v_{i_4}^4$ to $t$ is $i_4$. This implies that $P$ uses at least $n+4 + (i_4 - i_1)$ temporal edges. Since the length of $P$ is at most $n+4$, we get $i_4 \leq i_1$. We thus have $i_1 = i_4$. Now, the three temporal edges that $P$ uses between $v_{i_1}^1$ and $v_{i_4}^4=v_{i_1}^4$ must have the form $v_{i_1}^1 v_{i_2}^2, v_{i_2}^2 v_{i_3}^3$ and $v_{i_3}^3 v_{i_1}^4$ for some $i_2,i_3\in[n]$. By the way that $H_G$ is constructed, $\{v_{i_1}, v_{i_2}\}, \{v_{i_2} v_{i_3}\}$ and $\{v_{i_3}, v_{i_1}\}$ are undirected edges in $G$, i.e., $v_{i_1}, v_{i_2}, v_{i_3}$ forms a triangle in $G$.
This concludes the proof for delay zero.

\smallskip
For delay one, the proof is very similar.
We include it for the sake of completeness. 
To obtain an interval temporal graph $H_G'$ with uniform delay one, we modify the temporal edges of $H_G$ as follows. All delays are set to one. All permanent edges are present during interval $[1,3n+3]$ instead of $[1,n]$. For each undirected edge $\{v_i,v_j\}\in E$ with $i<j$, the symmetrical edges $v_i^1v_j^2$ and $v_j^2v_i^1$ are present at point $n+i$ while the symmetrical edges $v_j^3v_i^4$ and $v_i^4v_j^3$ are present at point $n+i+2$. The proof similarly follows from the following claim.

\begin{claim*}
    There is a triangle in $G$ if and only if there exists a temporal path in $H_G'$ from $s$ to $t$ using at most $n+4$ edges.
\end{claim*}

If $G$ contains a triangle $v_{i_1},v_{i_2},v_{i_3}$ with $i_1<i_2<i_3$, we can again form a temporal path of length $n+4$ in $H_G'$ starting from $s$ at time $0$, reaching $v_{i_1}^1$ within $V_1$ at time $n-i_1+1$, waiting in $v_{i_1}^1$ until time $n+i_1$ following the temporal edges $e_1=(v_{i_1}^1, v_{i_2}^2, n+i_1, n+i_1, 1), e_2=(v_{i_2}^2, v_{i_3}^3, 1, 3n+3, 1), e_3=(v_{i_3}^3, v_{i_1}^4, n+i_1+2, i_1, 1)$ at times $n+i_1,n+i_1+1, n+i_1+2$ respectively, reaching $v_{i_1}^4$ at time $n+i_1+3$, reaching $v_1^4$ within $V_4$ at time $n+2i_1+2$, and arriving in $t$ at time $n+2i_1+3$.

Conversely, suppose that $P$ is a temporal $st$-path in $H_G'$ of length at most $n+4$. Its first temporal edge leading outside $V_1\cup\{s\}$ must be $e=(v_{i_1}^1, v_{i_2}^2, n+i_1, n+i_1, 1)$ for some $i_1\in [n]$. Its last temporal edge with a vertex outside $V_4\cup\{t\}$ must be $g=(v_{i_3}^3, v_{i_4}^4, n+i_4+2, i_1, 1)$ for some $i_4\in [n]$. As the minimum length of path from $V_2$ to $V_3$ is one in the underlying graph of $H_G'$, there is at least one temporal edge $f$ between $e$ and $g$ in $P$. As $e$ and $g$ must be traversed at times $n+i_1$ and $n+i_4+2$ respectively, we get $n+i_1+2\le n+i_4+2$ since $P$ is a temporal path and $e,f$ have delay one. We thus have $i_1\le i_4$. Again, the distance from $s$ to $v_{i_1}^1$ in the underlying graph of $H_G'$ is $n-i_1+1$, the distance from $v_{i_1}^1$ to $v_{i_4}^4$ is at least 3, and the distance from $v_{i_4}$ to $t$ is $i_4$. This implies that $P$ has length at least $n+4+(i_4-i_1)\le n+4$. We thus have $i_1=i_4$ and $G$ contains the triangle $v_{i_1}, v_{i_2}, v_{i_3}$.
This concludes the proof for delay one.
\end{proof}

\section{Complexity of fastest temporal path}

\subsection{Lower bounds for fastest temporal path}

 \begin{theorem}\label{thm:hardness-apsp}
    If, for every $\varepsilon > 0$, there is no ${\cal O}(n^{3-\varepsilon})$-time algorithm for detecting a negative triangle in a weighted graph, then for every $\eps' > 0$, there
    is no ${\cal O}((nM)^{1-\eps'})$-time algorithm for computing a fastest temporal $st$-path in a temporal graph, even if 
    the temporal graph is undirected.
\end{theorem}

\begin{proof}
    The {\sc Negative Triangle Detection} problem asks for the existence in an edge-weighted undirected graph of a triangle whose total weight is negative.
    We now present an ${\cal O}(n+m)$-time reduction from {\sc Negative Triangle Detection} to the problem of computing a fastest temporal path in a temporal graph.
    For that, let $G=(V,E,w)$ be an arbitrary edge-weighted undirected graph with $n$ nodes and $m$ edges where $w:E\rightarrow\mathbb{Z}$ assigns a weight $w_{e}$ to each edge $e\in E$.
    Without loss of generality, $V = \{0,1,2,\ldots,n-1\}$ and $G$ is loopless.
    Let $T = 2n\max_{e \in E} |w_e|$.
    We construct  from $G$ an undirected temporal graph $H_G$ as follows:
    \begin{itemize}
        \item The vertex set is $\{s,t\} \cup V_1 \cup V_2 \cup V_3$ where $V_1,V_2,V_3$ are disjoint copies of $V$ and $s,t$ are two additional vertices. 
        For every $v \in V$, let $v_1 \in V_1, \ v_2 \in V_2, \ v_3 \in V_3$ denote its respective copies. $V_1, V_2, V_3$ are independent sets.
        \item The underlying edge set includes $\{ sv_1,v_3t : v \in V\} \cup \{ u_1v_2, v_1u_2, u_2v_3, v_2u_3 : \{u,v\} \in E\}$ (and symmetrical edges).
        \item For every $\{u,v\} \in E$, {edges $u_1v_2,v_1u_2,u_2v_3,v_2u_3$ and their symmetrical counterparts $v_2u_1,u_2v_1,v_3u_2,u_3v_2$} are all present during $[0,2mT]$ with delay $T/2 + w_{\{u,v\}}$. 
        \item Finally, let $E = \{e_1,e_2,\ldots,e_m\}$ be an enumeration of the edge set $E$. For each $i\in [m]$, let $u$ and $v$ be the endpoints of $e_i$ such that $u < v$. 
        {Then, the symmetrical edges $su_1,u_1s$ are both present at point $2(i-1)T$ with delay $T/2+w_{e_i}$, and the symmetrical edges $v_3t,tv_3$ are both present during $[2(i-1)T,2iT]$ with delay $T/2$.} 
    \end{itemize}
    
\begin{figure}[!ht]
    \centering
    \begin{tikzpicture}[scale=1.2, every node/.style={circle, draw, minimum size=8mm, font=\small}, >=Stealth]

        \node (s) at (-1, 4) {$s$};
        \node (t) at (9, 4) {$t$};

        \node (w1) at (1.5, 6) {$x_1$};
        \node (v1) at (1.5, 4) {$v_1$};
        \node (u1) at (1.5, 2) {$u_1$};

        \node (w2) at (4, 6) {$x_2$};
        \node (v2) at (4, 4) {$v_2$};
        \node (u2) at (4, 2) {$u_2$};

        \node (w3) at (6.5, 6) {$x_3$};
        \node (v3) at (6.5, 4) {$v_3$};
        \node (u3) at (6.5, 2) {$u_3$};

        \draw[thick, dotted, rounded corners, blue!50] (1.1, 6.5) rectangle (1.9, 1.5);
        \draw[thick, dotted, rounded corners, blue!50] (3.6, 6.5) rectangle (4.4, 1.5);
        \draw[thick, dotted, rounded corners, blue!50] (6.1, 6.5) rectangle (6.9, 1.5);

        \node[above=0.2cm, font=\large, text=blue!50, draw=none] at (1.5, 6.5) {$V_1$};
        \node[above=0.2cm, font=\large, text=blue!50, draw=none] at (4, 6.5) {$V_2$};
        \node[above=0.2cm, font=\large, text=blue!50, draw=none] at (6.5, 6.5) {$V_3$};

        \foreach \i/\j in {u/v, v/w, w/v, v/u, w/u} {
            \draw[thick, blue] (\i1) -- (\j2);
        }
        \draw[thick, blue] (u1) -- 
        node[midway, near end, shift={(0.8,0)}, draw=none, font=\large] {$\frac{T}{2}\!+\!w_{\{u,x\}}$}
        (w2);

        \foreach \i/\j in {u/w, u/v, v/w, v/u, w/u} {
            \draw[thick, blue] (\i2) -- (\j3);
        }
        \draw[thick, blue] (w2) -- 
        node[midway, near start, shift={(0.4,0.5)}, draw=none, font=\large] {$\frac{T}{2}\!+\!w_{\{x,v\}}$}
        (v3);

        \draw[thick, dashed] (s) -- (w1);
        \draw[thick, dashed] (s) -- (v1);
        \draw[thick, dashed] (s) -- 
        node[midway, left, shift={(0.4,-0.2)}, draw=none, font=\large] {$\frac{T}{2}+ w_{\{u,v\}}$}
        (u1);

        \draw[thick, dashed] (w3) -- (t);
        \draw[thick, dashed] (v3) -- 
        node[midway, above, shift={(-0.2,-0.15)}, draw=none, font=\large] {$\frac{T}{2}$}
        (t);
        \draw[thick, dashed] (u3) -- (t);

    \end{tikzpicture}
    \caption{
    Part of the underlying graph of $H_G$ corresponding to a triangle $u,v,x$ contained in $G$ assuming that the undirected edge $e_i\in E$ is $\{u,v\}$ with $u<v$. Recall that each undirected edge $\{y,z\}$ of the figure corresponds to two symmetric edges $yz$ and $zy$ of $H_G$. Here, the given labels indicate the delays of the corresponding temporal edges appearing during interval $[2(i-1)T,2iT)$. Plain blue edges are present during interval $[0,2mT]$. Dashed edges are more restricted: in particular, edges $su_1$ and $u_1s$ are present at point $2(i-1)T$, while edges $v_3t$ and $tv_3$ are present during interval $[2(i-1)T,2iT]$. Note that these edges  can have other appearances depending on other neighbors of $u$ and $v$ in $G$. 
    }
    \label{fig:negative_triangle_no_weights}
\end{figure}
    
    Note that $H_G$ has $3n+2$ vertices and $M = 12m$ temporal edges (counting also symmetrical temporal edges), and the construction thus takes linear time. 
    Figure~\ref{fig:negative_triangle_no_weights} shows a schematic view of the resulting temporal graph.
    
    We claim that there exists a temporal $st$-path in $H_G$ with duration less than $2T$ if and only if $G$ contains a negative triangle.
    For that, let us fix a temporal $st$-path $P$ of minimum duration. 
    By the construction of $H_G$, its first temporal edge is present at time $2(i-1)T$ for some $i\in [m]$. Let $u,v$ be the endpoints of $e_i$ so that $P$ starts with edge $su_1$, i.e. $e_i=\{u,v\}$ with $u<v$. If the duration of $P$ is less than $2T$, then $P$ ends with the only edge with head $t$ appearing in $[2(i-1)T,2iT)$, that is edge $v_3t$. 
    Let $k < n$ be the number of edges of $P$. The duration of $P$ must be at least $k(T/2 -  \max_{e \in E} |w_e|) = (k - k/n)T/2 > (k-1)T/2$ while also being less than $2T$. We thus have $k \le 4$.
    Furthermore, 
    we must have $k=4$ since the distance between $s$ and $t$ is at least 4 in the underlying graph of $H_G$.
    In this situation, there must exist a vertex $x$ such that the underlying edges of $P$ are $su_1,u_1x_2,x_2v_3,v_3t$, implying that $u,x,v$ is a triangle of $G$.
    As the temporal edges connecting $su_1$ and $v_3t$ have delays  $T/2+w_{e_i}$ and $T/2$ respectively, the duration of $P$ is at least $2T + w_{e_i}+w_{\{u,x\}}+w_{\{x,v\}}$, which is less than $2T$ if and only if the triangle $u,x,v$ is negative
    (see the corresponding path in Figure~\ref{fig:negative_triangle_no_weights} for an illustration).
    Conversely, if $u,x,v$ is a negative triangle, then up to reordering we can assume that its edges are $e_1,e_2,e_3$. Then, there exists a temporal $st$-path that starts at time $0$, has no waiting time, and has duration less than $2T$.
    Therefore, the claim is proved.
\end{proof}

{\bf Remarks.} Known reduction of what to negative triangle detection usually requires weights in the range $\set{-n^c,\ldots,n^c}$ for some large enough integer $c$~\cite{WilliamsW2018}. It yields an $\Omega((nm)^{1-\eps})$ bound in the regime $m=\Theta(n^2)$, implying that our lower-bound holds in the regime where the underlying graph has $\Theta(n^2)$ edges. 

\smallskip

The following hardness result is inspired from that of Theorem~\ref{thm:hardness-apsp}, but it requires some adjustments (and a different complexity hypothesis) in order to cope with uniform delay one.

\begin{theorem}\label{thm:hardness-triangle}
    If, for every $\varepsilon > 0$, there is no combinatorial ${\cal O}(n^{3-\varepsilon})$-time algorithm for detecting a triangle in an undirected graph, then for every $\eps' > 0$, there is no combinatorial ${\cal O}((nM)^{1-\eps'})$-time algorithm for computing a fastest temporal path in a temporal graph, even if 
    the temporal graph is undirected and has uniform delay one.
\end{theorem}

\begin{proof}
    Let $G=(V,E)$ be an arbitrary unweighted graph.
    Without loss of generality, $V = \{0,1,2,\ldots,n-1\}$ and $G$ is loopless.
    Let $N$ be some large enough constant (say, $N=10$).
    The undirected temporal graph $H_G$ is constructed from $G$ as follows:
    \begin{itemize}
        \item The vertex set is $\{s,t\} \cup V_1 \cup V_2 \cup V_3$ where $V_1,V_2,V_3$ are disjoint copies of $V$. 
        For every $v \in V$, let $v_1 \in V_1, \ v_2 \in V_2, \ v_3 \in V_3$ denote its copies.
        \item The underlying edge set includes $\{ sv_1,v_3t : v \in V\} \cup \{ u_1v_2, u_2v_3 : \{u,v\} \in E\}$ (and symmetrical edges). 
        \item For every $\{u,v\} \in E$, edges $u_1v_2,u_2v_3$ and their symmetrical counterparts $v_2u_1,v_3u_2$ are all present during $[0,n\cdot N+3]$. 
        \item For every $v \in V$, the symmetrical edges $sv_1,v_1s$ are both present at point 
        $v \cdot N$, 
        while the symmetrical edges $v_3t,tv_3$ are both present at point $u \cdot N+3$ for each $u \in N(v)$. 
        \item Finally, all temporal edges have delay one.
    \end{itemize}
    Note that we can construct $H_G$ from $G$ in ${\cal O}(n+m)$ time.

    We claim that $G$ contains a triangle if and only if the shortest duration for reaching $t$ from $s$ equals $4$.
    Indeed, if $u,v,w$ is a triangle of $G$, then there exists a temporal path starting at time $u \cdot N$, going by vertices $s,u_1,v_2,w_3,t$, with zero waiting time and total duration $4$.
    Conversely, assume the existence of a temporal path $P$ from $s$ to $t$, with total duration at most $4$.
    As edges have delay 1, it must have 4 edges at most, which is the least it can have as $V_1,V_2,V_3$ are vertex separators of the underlying graph (i.e. the removal of any $V_i$ for $i\in [3]$ disconnects $s$ from $t$).
    The underlying path of $P$ must thus be $s,u_1,v_2,w_3,t$ for some triple $u,v,w\in V$ satisfying $\{u,v\},\{v,w\}\in E$. This path has
    $su_1$ and $w_3t$ as first and last edges respectively.
    Then, the starting time of $P$ must be $u \cdot N$ and edge $w_3t$ must be present at time $u \cdot N+3$ which happens only when $u\in N(w)$.
    Since $G$ is loopless, $u,v,w$ are pairwise distinct.
    Therefore, $G$ contains the triangle $u,v,w$. This proves the claim.
\end{proof}

Note that the above proof does not hold with delay zero. If we consider zero delays in the above reduction, then the fastest path from $s$ to $t$ has duration 3 if and only if $G$ contains an odd cycle which is an easier problem than detecting a triangle (it can be tested in linear time). However, if we construct a directed temporal graph $H'$ where all underlying edges are directed from $s$ to $V_1$, from $V_1$ to $V_2$, from $V_2$ to $V_3$ and from $V_3$ to $t$, then it holds that $H'$ contains a temporal $st$-path of duration 3 if and only if $G$ contains a triangle.
We can thus state the following Theorem, its proof is almost identical to that of Theorem~\ref{thm:hardness-triangle} and is omitted.



\begin{theorem}\label{thm:hardness-triangle-dir}
    If, for every $\varepsilon > 0$, there is no combinatorial ${\cal O}(n^{3-\varepsilon})$-time algorithm for detecting a triangle in an undirected graph, then for every $\varepsilon' > 0$, there is no combinatorial ${\cal O}((nM)^{1-\varepsilon'})$-time algorithm for computing a fastest temporal path in a directed temporal graph with uniform delay zero.
\end{theorem}

\subsection{Fastest temporal path in an undirected temporal graph with uniform delay zero}

We obtain an algorithm for computing a fastest temporal path by solving the more complex profile problem.
Given a temporal graph $G$ and a pair of distinct vertices $s$ and $t$ in this graph, we  define the \emph{$st$-profile $P_{st}$} as the function $P_{st}(\tau)$ associating each departure time $\tau$ from $s$ to the earliest arrival time in $t$; in other words, $P_{st}(\tau)$ is the minimum arrival time of any temporal path leaving $s$ at a time no earlier than $\tau$. This function is non-decreasing and piece-wise linear. Furthermore, the slope of each linear piece is either zero or one. The reason is that we can define a profile function $E_{uv}$ for each underlying edge $uv$. The slope is one in intervals where the edge is present: for each temporal edge $(u,v,\tau_1,\tau_2,\delta)$, the earliest arrival time in $v$ is $\tau+\delta$ for $\tau\in [\tau_1,\tau_2]$. The slope is zero in intervals where it is not present, as the earliest arrival time is $\tau_1+\delta$ for $\tau<\tau_1$ if the next presence of $uv$ is during interval $[\tau_1,\tau_2]$ with delay $\delta$. Similarly, the profile function $f_P$ of a path $P=v_1,\ldots,v_k$ can be obtained by composition of $E_{v_1v_2},\ldots, E_{v_{k-1}v_k}$, resulting in a polyline whose slopes are obtained by multiplying slopes of the composing function and must be in $\set{0,1}$. Finally, if $\cal P$ denotes the set of all paths from $s$ to $t$ in the underlying graph, we have $P_{st}=\min_{P\in \cal P} f_P$ which is thus also a polyline with slopes in $\set{0,1}$. See \cite{DehneOS2012} for more details.

One way to represent this function is by a collection of triples $\{(\alpha_i,\beta_i,s_i)\}_{i=1}^{k}$, where $(\alpha_i)_{i=1}^k$ is non-decreasing, $\beta_i = P_{st}(\alpha_i)$ and for every $\tau \in (\alpha_{i-1}, \alpha_i]$, we have $P_{st}(\tau) = \beta_i + s_i(\tau - \alpha_i)$ where we define $\alpha_0=-\infty$. The minimum collection of such triples is what we call a minimal representation of the $st$-profile. Note that when $s=t$, the profile is the identity function and requires a specific representation. For a detailed study of a general algorithm to find the profile of a pair of vertices in an interval temporal graph with arbitrary delays in $O((n\log n + m)M)$ time, we refer the reader to \cite{DehneOS2012}.

\begin{theorem}\label{thm:zero-delay}
    Given an undirected interval temporal graph $G$ with no isolated vertices and $M$ temporal edges having all zero delay, and two vertices $s$ and $t$, it is possible to compute in $\widetilde{\cal O}(M)$ time a fastest temporal path from $s$ to $t$. Furthermore, a minimal representation of the $st$-profile $P_{st}$ can be obtained with the same complexity.
\end{theorem}

To prove the above theorem, we propose a profile algorithm that, given the list of temporal edges of an interval temporal graph $G=(V,E)$ and two distinct vertices $s,t\in V$, computes a representation of the profile function $P_{st}$ from $s$ to $t$. The idea of the algorithm is to perform a time scan of the beginnings and ends of all temporal edges while maintaining connected components for edges present at each time $\tau$ considered. We also maintain for each connected component $c$ the last departure time $LDT[c]$ allowing us to reach it at time $\tau$ from $s$. In other words, $LDT[c]$ is the maximum time $\lambda$ such that there exists a temporal path from $s$ to any vertex $w$ in the component that starts at time $\lambda$ and arrives no later than time $\tau$. Note that this definition does not depend on $w$: such a temporal path to $w$ can be extended to a temporal walk that arrives at any other vertex of the component at time $\tau$ since all edges of the component are present at time $\tau$ and have zero delay.
(And a temporal path can be extracted from that temporal walk by waiting instead of following any loop.)

For that purpose we consider each temporal edge $(u,v,\tau_1,\tau_2,0)$ as two \emph{events}: a \emph{beginning event} (resp. \emph{ending event}) which occurs when the edge begins at time $\tau_1$ (resp. ends at time $\tau_2$) and which is defined as the quintuple $(u,v,\tau_1,\text{begin},\tau_2)$ (resp. $(u,v,\tau_2,\text{end},\tau_2)$). As we consider an undirected temporal graph $G$, we generate only one beginning event and one ending event for each pair of symmetrical temporal edges. We assume that the temporal graph $G$ is given as a sorted event list, ordered by the time of the events (i.e., the third coordinate). If several events share the same time, beginning events appear before ending events, breaking ties arbitrarily among events of the same type. Note that this list can easily be obtained from the list of temporal edges of $G$ in $O(M\log M)$ time using sorting. 


To maintain connected components as we scan edge events, we could use a dynamic connectivity algorithm~\cite{HolmLT2001,Eppstein1994} (either fully dynamic or offline). However, we note that can we can rely on the simpler dynamic tree algorithm of \cite{SleatorT1983}. The reason is that the ending times of edges are known in advance. We can thus maintain for each connected component a spanning tree consisting of edges with greatest ending times as we now explain. More precisely, we maintain a maximum-cost spanning tree of each connected component, where the cost of an edge is defined as the ending time of the corresponding temporal edge.
Indeed, the data-structure of \cite{SleatorT1983} maintains a collection of dynamic trees on a fixed set $V$ of vertices where each edge of a tree is associated to a cost.
In particular, it allows to retrieve in logarithmic time the edge with minimum cost along the path from a node to the root of its tree. When adding an edge $\{u,v\}$, if $u$ and $v$ are already in the same tree $T$, we can thus find the edge $\{u',v'\}$ with smallest ending time along the cycle of $T\cup \{u,v\}$ and cut $T$ by deleting $\{u',v'\}$ before linking $u$ and $v$ (if $\{u',v'\}=\{u,v\}$, we do nothing). If $u$ and $v$ are not in the same tree, we link the two trees according to \cite{SleatorT1983}. We let $CC\_AddEdge(\{u,v\},\tau_2)$ denote the resulting procedure for adding an edge $\{u,v\}$ with ending time $\tau_2$. We let $CC\_RemoveEdge(\{u,v\})$ denote the procedure that cuts the tree containing $u$ and $v$ according to \cite{SleatorT1983} if one vertex is parent of the other (and does nothing otherwise). We further let $CC(u)$ denote the procedure that returns the ID of the root of the tree containing $u$ according to \cite{SleatorT1983}. We use this ID (a number between 1 and $n$) to identify the component spanned by the tree, so that the expression $CC(u)=CC(v)$ allows us to test whether $u$ and $v$ are in the same connected component. All these procedures can be implemented in $O(\log n)$ time~\cite{SleatorT1983}. We let $CC\_Init(V)$ denote the procedure that creates $n$ single-node trees, one for each vertex in $V$ (in linear time).  Processing the event list in order, we can thus update the dynamic trees representing the connected components of the graph of edges present at the time of each event. 

\begin{algorithm}[t]
    \Input{A temporal graph $G$ given by a sorted list $E$ of edge events with vertex set $V=[n]$, a source node $s\in V$ and a target node $t\in V\setminus\set{s}$.}
    \Output{The $st$-profile.}
    $Prof:=\emptyset$ \tcp*{initial $st$-profile}
    $CC\_Init(V)$ \tcp*{initially empty graph}
    $LDT[c]:=-\infty$ for all $c\in [n]$ \tcp*{last departure time of a walk from $s$ to any vertex in a connected component $c$}
    \For{$e=(u,v,\tau,evt,\tau_2)\in E$}{
        $LDT[CC(s)]:= \tau$  \tcp*{$s$ can reach its component at time $\tau$} \label{lin:source}  
        $\lambda_t:=LDT[CC(t)]$ \tcp*{last departure time to reach $t$}
        $st\_connected:= CC(s) = CC(t)$\tcp*{true if $s$ and $t$ are connected}
        \uIf(\tcp*[h]{undirected edge $\{u,v\}$ appears}){$evt=\text{begin}$}{
            $\lambda := \max\set{LDT[CC(u)], LDT[CC(v)]}$\label{lin:max}\\
            $CC\_AddEdge(\{u,v\},\tau_2)$\\
            $c := CC(u)$ \tcp*{connected component containing $u$ and $v$}
            $LDT[c]:=\lambda$ \label{lin:merge} \\
            \lIf{$CC(t)=c$ 
                 and $\lambda>\lambda_t$}{$Prof:=Prof\cup\set{(\lambda, \tau, 0)}$
                 }\label{lin:plateau}
        }\Else(\tcp*[h]{undirected edge $\{u,v\}$ disappears}){
            $c := CC(u)$ \tcp*{connected component containing $u$ and $v$}
            $\lambda := LDT[c]$\\
            $CC\_RemoveEdge(\{u,v\})$\\
            $LDT[CC(u)]:=\lambda$\\
            $LDT[CC(v)]:=\lambda$\\
        }
        \lIf{$st\_connected$}{$Prof:=Prof\cup\set{(\tau, \tau, 1)}$
        \label{lin:connected}} 
        }
    \Return $Prof$
    \caption{One-to-one profile.}
    \label{alg:profile}
\end{algorithm}

The update of last departure times is then rather simple. When two components merge because of the appearance of an edge $uv$ at time $\tau$, the last departure time $LDT[c]$ of the resulting component $c$ is set to the maximum of the last departure times of the components of $u$ and $v$ before merging since a temporal path reaching one component can now be extended to reach any node in the other component. When a component splits because of the disappearance of an edge $uv$ at time $\tau$, the two new components get same last departure time as the component before splitting since temporal paths reaching that component arrive at time $\tau$ or before, and waiting is then possible. Special care has to be taken concerning the connected component of the source $s$ for which the last departure time is always the time $\tau$ of the current event. Obtaining the last departure time $\lambda$ of the connected component of $t$ after an event at time $\tau$ basically indicates that $P_{st}(\lambda)\le\tau$. We will prove that  when a higher value of $\lambda$ is observed for $t$, $\tau$ corresponds to the earliest arrival time when departing at $\lambda$ and we then have $P_{st}(\lambda)=\tau$. Note that $\tau-\lambda$ is then the duration of the corresponding temporal path from $s$ to $t$. Each time an event affects the connected component of the target $t$, we update a list $Prof$ of triples accordingly so that it represents the profile $P_{st}$ up to that event.
%
See the pseudo-code of Algorithm~\ref{alg:profile} for more details.

\begin{proposition}\label{prop:profile}
    Given an undirected interval temporal graph $G$ with with no isolated vertices and $M$ temporal edges having all zero delay, and two vertices $s$ and $t$, Algorithm~\ref{alg:profile} computes a representation of the $st$-profile in $O(M(\log M+\log n))$ time.
\end{proposition}

First note that Theorem~\ref{thm:zero-delay} easily follows from Proposition~\ref{prop:profile} since the duration of a fastest temporal $st$-path as well as the starting time $\lambda$ of such a fastest temporal path can easily be obtained by scanning the $st$-profile in $O(M)$ time to obtain the triple $(\lambda,\tau,s)$ minimizing $\tau-\lambda$. We can then find a fastest temporal $st$-path departing at time $\lambda$ by computing an earliest-arrival $st$-path departing at time $\lambda$ in $O(M\log M)$ time. Indeed, after  an $O(M\log M)$-time pre-processing for obtaining for each underlying edge $uv$ the sorted list of temporal edges connecting $uv$, we can compute such a temporal path through a temporal version of Dijkstra's algorithm~\cite{Dreyfus1969,BuiXuanFJ2003} in $O(m\log M + n\log n)$ time. Note that the $\log M$ factor accounts for the time required to find with binary search the first temporal edge connecting $uv$ that can be traversed at a certain time $\tau$ for a given neighbor $v$ (in the underlying graph) of a node $u$ reached at time $\tau$. The pre-processing also includes removing the overlaps between the temporal edges corresponding to various appearances of the same edge, see~\cite{jain2022algorithms} for more details.
As a final remark, a minimal representation of the profile can easily be obtained from the computed one through a linear-time post-processing scan where we merge two consecutive triples if they correspond to the same line.

\begin{proof}[Proof of Proposition~\ref{prop:profile}]
    Considering all events up to an event $e\in E$, let $G^e$ be the temporal graph with the same set of vertices as $G$ and all temporal edges $(u,v,\tau_1,\tau_2)$ of $G$ such that the event $(u,v,\tau_1,\text{begin},\tau_2)$ is not after $e$ in the event list $E$. 
    The temporal edges of $G^e$ whose ending events come after $e$ in $E$, are called \emph{unclosed temporal edges} as their ending events have not been processed yet. 
    We let $\tau^e$ denote the time at which event $e$ occurs.
    We define the last departure time $\lambda^e(w)$ from $s$ to reach $w$ at $e$ or before in $G^e$ as the maximum time $\lambda$ such that there exists a temporal $sw$-walk departing at time $\lambda$ in $G^e$ and arriving no later than $\tau^e$. 
    Similarly, we use $CC^e$, $LDT^e$ and $Prof^e$ to refer to the values stored in $CC$, $LDT$ and $Prof$ respectively after the algorithm has processed the event $e$. We aim at proving the following claim.

    \begin{claim}\label{clm:algo-ind-bis}
        For each event $e$, we have $LDT^e[CC^e(w)] = \lambda^e(w)$.
    \end{claim}
    
    We first note that this claim allows us to prove that Algorithm~\ref{alg:profile} correctly identifies a triple $(\lambda,\tau,s)$ corresponding to a fastest temporal $st$-path with duration $\tau-\lambda$. Consider a fastest temporal $st$-path $P^*$ and let $\lambda^*$ and $\tau^*$ denote its departure time from $s$ and arrival time in $t$ respectively. Among the suffixes of $P^*$ with temporal edges present at time $\tau^*$, consider the last beginning event $e$ of one of these temporal edges and consider its time $\tau^e\le\tau^*$. As $\tau^e < \tau^*$ would lead to a faster temporal $st$-path, we must have $\tau^e=\tau^*$.  As the suffix from the temporal edge associated to $e$ is present at time $\tau^e$ by the choice of $e$ and $t$ is in the connected component of the suffix when processing $e$, the claim implies $LDT^e[CC^e(t)] = \lambda^e(t)$ and $\lambda^e(t)=\lambda^*$ by definitions of $\lambda^e(t)$ and  $P^*$.
    If $\lambda^*>\lambda_t$, the triplet $(\lambda^*,\tau^*,0)$ is added to the profile at Line~\ref{lin:plateau}. Otherwise, a triplet $(\lambda_t,\tau',s)$ was added previously with $\lambda^*\le\lambda_t$ and $\tau'\le \tau^*$, implying $\tau'-\lambda_t\le\tau^*-\lambda^*$. As $P^*$ is a fastest temporal path, we must have $\tau'-\lambda_t=\tau^*-\lambda^*$.
    In both cases, Algorithm~\ref{alg:profile} correctly identifies a triple $(\lambda,\tau,s)$ corresponding to a fastest temporal $st$-path.
    Additionally, Claim~\ref{clm:algo-ind-bis} also implies that each triple $(\lambda,\tau,0)$ added to $Prof$ corresponds to a temporal path starting from $s$ at time $\lambda$ and arriving in $t$ at time $\tau$.
    Moreover, a triple $(\tau,\tau,1)$ is added to $Prof$ only if $s$ and $t$ are in the same connected component at time $\tau$. We can thus conclude that the minimum $\tau-\lambda$ over all triples $(\lambda,\tau,s)$ of $Prof$ is the duration of a fastest temporal $st$-path.
    
    \smallskip
    \textit{Proof of Claim~\ref{clm:algo-ind-bis}.}
    We prove Claim~\ref{clm:algo-ind-bis} by induction on the number of events processed. It relies mainly on the correctness of the dynamic connected component procedures which implies that $CC^e(w)\in [n]$ uniquely identifies the connected component of $w$ in the graph induced by the edges that are present in $G$ at time $\tau^e$ and whose ending events are after $e$ in $E$. The ordering of $E$, with all beginning events at a given time $\tau$ preceding all ending events at time $\tau$, implies that these temporal edges are exactly the unclosed temporal edges of $G^e$.
    In particular, when two vertices $u$ and $v$ are in the same connected component (i.e. $CC^e(u)=CC^e(v)$), they are connected through unclosed temporal edges of $G^e$. Note that this implies the existence of an instantaneous temporal $uv$-path that uses only unclosed temporal edges which are all traversed at the same time $\tau^e$.
    
    Initially, we consider an empty temporal graph $G^0$ and the claim is clearly satisfied for any vertex $w\not=s$. 
    Now, suppose that Claim~\ref{clm:algo-ind-bis} is satisfied for any vertex $w\not=s$ before scanning an event $f$ at time $\tau=\tau^f$. Note that the algorithm first sets $LDT[CC(s)]:=\tau$, so that the claim $LDT^f[CC^f(w)]=\lambda^f(w)$ is satisfied for $w=s$ and all vertices $w$ in $CC^f(s)$. Now consider any vertex $w$ outside of $CC^f(s)$. In that case, any temporal $sw$-path uses some temporal edge ending before $\tau^f$, implying $\lambda^f(w)\le\tau^{e}$ where $e$ denotes the event preceding $f$ (we use $e=0$ and $\tau^{e}=-\infty$ if $f$ is the first event). We distinguish two cases.
    
    \textit{Ending event case.}
    If $f$ corresponds to the disappearance of an edge $uv$, we have $G^f=G^{e}$, implying $\lambda^f(w)\ge \lambda^{e}(w)$. As we have $\lambda^f(w)\le\tau^{e}$ (since $w$ is outside of $CC^f(s)$), the corresponding temporal $sw$-path exists in $G^{e}$ and we have $\lambda^f(w)= \lambda^{e}(w)$.
    The updates in that case ensure that $LDT[CC(w)]$ keeps the same value (after the possible change of $CC(w)$) so that $LDT^f[CC^f(w)]=\lambda^{e}(w)=\lambda^f(w)$ still holds. 
    
    \textit{Beginning event case.}
    If $f$ corresponds to the appearance of an edge $uv$,
    $G^{f}$ differs from $G^{e}$ by the addition of the two temporal edges $f'=(u,v,\tau^f,\tau_2)$ and $f''=(v,u,\tau^f,\tau_2)$ associated to $f$. We thus have $\lambda^f(w)\ge \lambda^{e}(w)$, and $\lambda^f(w) > \lambda^{e}(w)$ can only occur if any temporal $sw$-path with departure time $\lambda^f(w)$ uses $f'$ or $f''$. Consider such a temporal $sw$-path departing at time $\lambda^f(w)$ and assume without loss of generality that it traverses $f'$ from $u$ to $v$, implying $\lambda^f(w)\le \lambda^{e}(u)$ by definition of $\lambda^{e}(u)$.
    As we set $LDT[CC(w)]$ to $\max\set{LDT[CC(u)], LDT[CC(v)]}$ in that case, we then have $LDT^f[CC^f(w)]\ge LDT^{e}[CC^{e}(u)]=\lambda^{e}(w)$ by induction hypothesis. Combining this with the previous inequality, we get $LDT^f[CC^f(w)]\ge \lambda^f(w)$.
    We also have $\lambda^f(w)\ge LDT^f[CC^f(w)]$ since the induction hypothesis implies the existence of a temporal $sx$-path $P$ to a node $x$ in $\{u,v\}$ departing at $\max\set{LDT^{e}[CC^{e}(u)], LDT^{e}[CC^{e}(v)]}=LDT^f[CC^f(w)]$ and arriving at $\tau^e$ or before such that $P$ can be extended by an instantaneous temporal $xw$-path at time $\tau^f\ge\tau^e$. We thus finally get $LDT^f[CC^f(w)]= \lambda^f(w)$ in that case too.
    Claim~\ref{clm:algo-ind-bis} thus follows by induction.

    \medskip
    We now turn to the proof of the profile computation.
    We denote by $P_{st}^e$ the $st$-profile in $G^e$ up to $\tau^e$ (i.e. with support restricted to $(-\infty,\tau^e]$) and let $last(P_{st}^e)$ denote the maximum time $\tau\le\tau^e$ for which $P_{st}^e(\tau)$ is defined (by convention $last(P_{st}^e)=-\infty$ when $G^e$ does not contain any temporal path from $s$ to $t$). We aim at proving the following claim.

    \begin{claim}\label{clm:algo-ind}
        For each event $e$, $Prof^e$ is a representation of $P_{st}^e$.
    \end{claim}
    
    Note that applying this claim to the last event $l$ proves the correctness of Algorithm~\ref{alg:profile} as it ensures that $Prof^l$ is a representation of the $st$-profile in $G=G^l$. 
    The complexity of the algorithm directly follows from the  time complexities of the dynamic connected component procedures. 
    Proposition~\ref{prop:profile} will thus follow from the proof of this claim.
    
    \smallskip
    We first state the following about profile inclusion.

   \begin{claim}\label{clm:profile}
        For any two events $e$ and $f$ such that $e$ precedes $f$ in $E$, and any time $\lambda\le last(P_{st}^e)$, we have $P_{st}^f(\lambda)= P_{st}^e(\lambda)$.
    \end{claim}
    
    Since any temporal path in $G^e$ is also a temporal path in $G^f$, we obtain $P_{st}^f(\lambda)\le P_{st}^e(\lambda)$ for any time $\lambda\le last(P_{st}^e)$. We also have $P_{st}^e(\lambda)\le \tau^e$ since a temporal path arriving after $\tau^e$ in $G^e$ uses only unclosed temporal edges after $\tau^e$, and since these edges begin at $\tau^e$ or before we can transform the temporal path into one arriving at $\tau^e$. This implies that we have equality of profiles for $\lambda\le last(P_{st}^e)$ since no temporal path in $G^f$ can arrive before $P_{st}^e(\lambda)\le \tau^e$ as the possible additional temporal edges of $G^f$ begin at $\tau^e$ or later.  

    \smallskip
    We now prove Claim~\ref{clm:algo-ind} by induction on the number of events processed.
    Initially, we consider an empty temporal graph $G^0=(V,\emptyset)$ (with no temporal edge). As expected, $Prof$ is empty reflecting that there exists no temporal walk from $s$ to $t$ in $G^0$. 
    We assume the induction hypothesis for event $e$, i.e. $Prof^e$ is a representation of $P_{st}^e$. 
    Let us consider the next event $f$ and distinguish two cases.
    
    \textit{Ending event case.}
    First suppose that $f$ is an ending event $(u,v,\tau^f,\text{end},\tau^f)$. 
    Note that we have $G^e= G^f$ by definition. 
    By Claim~\ref{clm:profile}, we have $P_{st}^f(\lambda)=P_{st}^e(\lambda)$ for $\lambda\le last(P_{st}^e)$ so that $Prof^e$ is a valid representation of $P_{st}^f(\lambda)$ for $\lambda\le\tau^e$. Moreover, if $\lambda^f(t) \le \tau^e$,  $G^e=G^f$ implies $\lambda^f(t)= \lambda^e(t)\le last(P_{st}^e)$ and we have $P_{st}^f=P_{st}^e$. The only case when $P_{st}^f$ differs from $P_{st}^e$ is thus when $\lambda^f(t) > \tau^e$, that is when there exists a temporal $st$-path composed of unclosed temporal edges. In other words, this occurs when
    $CC^e(t)=CC^e(s)$ and $st\_connected$ is true. This implies that we then have $last(P_{st}^e)=\tau^e$ and $P_{st}^e(\tau^e)=\tau^e$. Moreover, $P_{st}^f$ is then the identity over $[\tau^e,\tau^f]$ as reflected by the addition of the triple $(\tau^f,\tau^f,1)$ at the end of $Prof$ in that case at Line~\ref{lin:connected}. $Prof^f$ is thus a valid representation of $P_{st}^f$.

    \textit{Beginning event case.}
    Now suppose that $f$ is a beginning event $(u,v,\tau^f,\text{begin},\tau_2)$ for the temporal edge $f'=(u,v,\tau^f,\tau_2)$ of $G$. Note that $G^f$ then differs from $G^e$ by having the two extra temporal edges $f'$ and its symmetrical temporal edge $f''=(v,u,\tau^f,\tau_2)$. 
    Consider the $st$-profile $P_{st}^f$ in $G^f$. It is the identity during interval $[\tau^e,\tau^f]$ when $CC^e(s)=CC^e(t)$ in correspondence with the addition of $(\tau^f,\tau^f,1)$ to $Prof$ at Line~\ref{lin:connected} in that case. Note that we then have $last(P_{st}^e)=\tau^e$ and $P_{st}^e(\tau^e)=\tau^e$.

    Now consider the case when $t$ is connected to $u$ or $v$, we have $\lambda_t=LDT^e[CC^e(t)]=\lambda^e(t)$ and $\lambda^f(t)=LDT^f[CC^f(u)]=\lambda$ by Claim~\ref{clm:algo-ind-bis}. If $\lambda>\lambda_t$, this is in correspondence with the addition of $(\lambda,\tau^f,0)$ at Line~\ref{lin:plateau} as a temporal $st$-path departing at time $\lambda$ must use $f'$ or $f''$ and cannot arrive before $\tau^f$. Note that we have $\lambda=\tau^f=\lambda_t$ when $CC^e(s)=CC^e(t)$ and the triple is not added in that case. Otherwise, $\lambda\le \lambda_t$ can only occur when $\lambda=\lambda_t=\lambda^f(t)=\lambda^e(t)$ since we have $\lambda^f(t)\ge\lambda^e(t)$ as any temporal path in $G^e$ is also a temporal path in $G^f$. We then have $\lambda^f(t)=last(P_{st}^e)$ implying $P_{st}^f=P_{st}^e$ which justifies to not modify $Prof$ in that case.

    Finally consider the case when $t$ is not in the connected component of neither $s$, nor $u$, nor $v$. We then have $P_{st}^f=P_{st}^e$ since a temporal $st$-path in $G^f$ cannot traverse $f'$ or $f''$ as they appear to late to reach $t$ after $\tau^f$. Any temporal $st$-path in $G^f$ is thus also a temporal $st$-path in $G^e$ with same departure time. This is indeed reflected by the fact that $Prof$ is not modified in that case, yielding to $Prof^f=Prof^e$. We finally obtain that $Prof^f$ is a valid representation of $P_{st}^f$ in all cases. This concludes our proof of Claim~\ref{clm:algo-ind} as well as the proofs of Proposition~\ref{prop:profile} and Theorem~\ref{thm:zero-delay}.
\end{proof}

\section{Conclusion}

We have presented non-trivial lower bounds showing a complexity gap between point temporal graphs and interval temporal graphs. They also show a complexity gap between the computation of a foremost temporal path and that of a fastest (resp. shortest) temporal path. As far as we know, these are the first results proving these gaps. Several questions arise from this work.

First, can we close the gap between our $\Omega(nM)$ lower bound for the computation of a fastest temporal $st$-path and the $\widetilde{\cal O}(mM)=\widetilde{\cal O}(n^2M)$ upper bound given by the best-known algorithm \cite{DehneOS2012}. Note that the gap can be large for dense underlying graphs with $m=\Theta(n^2)$. 

Combinatorial algorithms for computing a shortest path in a graph can usually be extended to compute one-to-all shortest paths. In the case of uniform zero delay, can we similarly hope to find an $\widetilde{\cal O}(M)$ algorithm for computing one-to-all fastest durations, i.e., given a source vertex $s$, an algorithm that computes the duration of a fastest temporal path from $s$ to each possible target vertex $t$?
Note that we cannot hope to turn Algorithm~\ref{alg:profile} into a one-to-all profile algorithm with the same complexity, since a profile can be of size $\Omega(M)$. But it can easily be turned into an $\widetilde{\cal O}(nM)$-time algorithm for one-to-all profiles. More generally, does this latter complexity also apply to interval temporal graphs with constant delays? Note that a positive answer would also apply to the computation of one-to-all fastest durations, and thus would probably also solve our first question.

Testing temporal connectivity might be related to the existence of a specific temporal spanner, that is a subset of temporal edges which are sufficient to preserve temporal connectivity. 
Note that our reduction for proving Theorem~\ref{thm:hardness-connectivity} produces a sparse temporal graph with $O(n)$ temporal edges, implying that the existence of a sparse spanner is not enough to enable temporal connectivity testing in subquadratic time.
We thus ask whether subquadratic time could be enabled by the existence of a specific well-structured temporal spanner. For example, a pivot vertex~\cite{CasteigtsPS2021}, which is a vertex that all vertices can reach by some time $\tau$ and which can reach all other vertices after $\tau$, enables 
such a temporal spanner. Thus, we ask whether such a pivot and an appropriate time $\tau$ can be found in subquadratic time when they exist.

Another natural extension of this work is to consider other problems beyond reachability and other complexity gaps such as P vs NP.

\bibliography{biblio}

\end{document}